\documentclass[letterpaper, 10 pt, conference]{ieeeconf}  

\IEEEoverridecommandlockouts                              
\usepackage{amsmath, amsfonts, amssymb}
\usepackage{algorithm, algorithmicx, algpseudocode}
\usepackage{graphicx}
\usepackage[english]{babel}
\usepackage[latin1]{inputenc}
\usepackage{fancyhdr}
\usepackage{yfonts}
\usepackage{mathrsfs}
\usepackage[rflt]{floatflt}
\usepackage{cite}
\usepackage{verbatim}
\usepackage{array}
\usepackage{float}
\usepackage{subfigure}

\newtheorem{proposition}{Proposition}

\newtheorem{lemma}{Lemma}
\newtheorem{assumption}{Assumption}

\newtheorem{remark}{Remark}

\newenvironment{proofof}{\noindent {\em Proof of }}{\hfill \hspace*{1pt}
	\hfill $\blacksquare$}

\newcommand{\tr}[1]{\mathrm{tr}{#1}}

\providecommand{\prt}[1]{\left( #1 \right)}

\providecommand{\abs[1]}{\left|#1\right|}

\providecommand{\quotes[1]}{``#1"}

\providecommand{\M}{{\mathcal{M}}}

\providecommand{\com[2]}{\begin{tt}[#1: #2]\end{tt}}
\usepackage{color}

\newcommand{\alex}[1]{{\color{black} #1}}
\newcommand{\alexb}[1]{{\color{black} #1}}

\newcommand{\jmh}[1]{{#1}}
\newcommand{\jmhnew}[1]{{\color{black} #1}}
\newcommand{\revjh}[1]{{\color{black} #1}}

\title{\LARGE \bf
Spectral Identification of Networks with Inputs
}

\author{Alexandre Mauroy and Julien M. Hendrickx
\thanks{This paper presents research results of the Belgian Network DYSCO (Dynamical Systems, Control, and Optimization), funded by the Interuniversity Attraction Poles Programme initiated by the Belgian Science Policy Office.}
\thanks {A. Mauroy is with Namur Center for Complex Systems (naXys) and Department of Mathematics, University of Namur, Belgium  {\tt\small alexandre.mauroy@unamur.be}}
\thanks {J. Hendrickx is with ICTEAM Institute,
Universit\'e catholique de Louvain, Belgium {\tt\small julien.hendrickx@uclouvain.be}}
}

\begin{document}

\maketitle
\thispagestyle{empty}
\pagestyle{empty}


\begin{abstract}
	We consider a network of interconnected  dynamical systems. Spectral network identification consists in recovering the eigenvalues of the network Laplacian from the measurements of a very limited number (possibly one) of signals. These eigenvalues allow to deduce some global properties of the network, such as bounds on the node degree.
	Having recently introduced this approach for autonomous networks of nonlinear systems, we extend it here to treat networked systems with external inputs on the nodes, in the case of linear dynamics. This is more natural in several applications, and removes the need to sometimes use several independent trajectories. We illustrate our framework with several examples, where we estimate the mean, minimum, and maximum node degree in the network. Inferring some information on the leading Laplacian eigenvectors, we also use our framework in the context of network clustering.
\end{abstract}

%
%
%
%
%
%

\section{INTRODUCTION}

The problem of inferring the structure of a network of dynamical systems is relevant in various fields such as biology (e.g. reconstructing regulatory networks from gene expression data), neuroimaging (e.g. revealing the structural organization of the brain), and engineering (e.g. localizing failures in power grids or computer networks), to list a few.

Most works aim at fully recovering the network structure. In the absence of prior knowledge, this can be proved to require measuring all nodes, which can be costly or \alex{even impossible given the current state of technology}. 
\alex{For example, it is obvious that} measuring all neurons in a brain is out of the question. A possible approach to solve this problem consists in working on a smaller \quotes{compressed} network whose nodes are precisely the measurement points. \alexb{But the topological properties of the identified compressed network may be very different from those of the original network.}

\alex{In this paper, we take the view that re-identifying the whole network is in some cases a too ambitious and costly goal, which is not always necessary. Instead, the aim of our spectral network identification approach is} to recover global information about the network structure \alex{from a small number (possibly one) of local measurements.} This is achieved by using the network Laplacian eigenvalues as proxy. From a practical point of view, we use the \alexb{Dynamic Mode Decomposition (DMD)} method \cite{Schmid, Tu} to extract the eigenvalues \alex{of the whole networked system} \alex{from the observation of a few signals, whether node trajectories or combinations of these.} We then deduce the Laplacian eigenvalues, or approximations of the Laplacian eigenvalues moments, depending on the size of the network. This spectral information provides bounds on various quantities of interest in the network (e.g. mean, minimum, maximum node degree).

We have developed the spectral network identification approach in \cite{MauroyHendrickx:17} for interconnected nonlinear systems without external inputs. The non-linearity made it impossible to work on the \quotes{system eigenvalues}, but we developed an approach similar to that just explained using the eigenvalues of the Koopman operator, which allows to represent a nonlinear system by a linear infinite-dimensional one. However, this did require to assume that the network would eventually synchronize. As a result of the synchronization and absence of external input, only a finite amount of information could be obtained from one experiment (i.e. one trajectory), so that it was necessary to perform several ones (i.e. re-setting the system at an arbitrary initial condition). Our method did for example allow to detect that a new node with lower degree has been added to a network, or to estimate the average node degree, solely by measuring the trajectory of one single node which needs not to be particularly representative of the typical nodes in the network. 


\emph{Contribution:} In this work we focus on \alex{networks of identical} linear interconnected systems, and extend our approach to take into account external inputs, \alex{exploiting} a recent variation of the DMD algorithm \cite{Proctor}. First, this allows to treat \alex{more general} systems where external excitation is present and cannot be removed. \alex{More importantly, inputs continuously excite the system, making} the experiments more informative, with the amount of obtained information growing over time. We can \alex{obtain fair results with a single trajectory, and we no longer assume the system to converge to a synchronized state.}
As a side contribution, we also show that we can infer some information on the leading Laplacian eigenvectors, providing a technique that can be used to cluster measured nodes. \alexb{Our results are illustrated in Section \ref{sec:large_moments}, where we (1) compute the average node degree of a large random network, (2) estimate the minimum and maximum node degree in the Polblogs network with opinion dynamics, and (3) detect clusters in a random graph with planted partitions.}


\paragraph*{Related approaches} Network identification problems have received increasing attention over the past years, and the topic is actively growing in nonlinear systems theory. See e.g. the recent survey \cite{Timme_review}. Many methods have been developed, exploiting techniques from various fields: velocities estimation \cite{Pikovsky_net_ident, Timme2_net_ident}, adaptive control \cite{Yu_net_estimation}, steady-state control \cite{Yu_Parlitz_steady_state}, optimization \cite{net_ident_opti}, compressed sensing \cite{Sauer_net_ident,net_ident_compressed_sens}, stochastic methods \cite{Ren_net_ident}, etc. These methods provide the structural (i.e. exact) connectivity of the underlying network and exploit to do so the dynamical nature of the individual units, which is often known, at least partially. As already mentioned, most of those work aim at fully recovering the network structure. 
\alexb{A few works have also been proposed to estimate the Laplacian eigenvalues of the network \cite{Banaszuk_eigenvalues, Franceschelli, Kibangou}. These methods yield decentralized algorithms, but impose a specific consensus dynamics at each node. In contrast, our method is not decentralized (except when only one node is mesured) but does not impose specific dynamics.}
Finally, we note a difference with an alternative line of works where the structure of the network is known, and the goal is to identify the dynamics taking place at the different nodes or edges, see e.g. \cite{Chiuso,Dankers,Gevers}.

%
%
%
%
%
%

\section{Relevance of Laplacian eigenvalues}\label{sec:intro_laplacian}

%

Consider a weighted network $G=(V,W)$, where $V$ is a set of $n$ nodes and $W_{ij}=[W]_{ij}>0$ is the weight of the edge $(j,i)$,  with in particular $W_{ij}=0$ in the absence of edge. \alex{In the case of unweighted networks, we consider that $W_{ij}=1$ if there is an edge $(j,i)$ and $W_{ij}=0$ otherwise.} The \emph{weighted in-degree} \alex{(or simply degree)} $d_i$ of node $i$ is the sum of the weights of edges arriving at $i$, i.e. $d_i = \sum_{j}W_{ij}$. The \emph{Laplacian} matrix $L$ of the network $G$ is defined by
$L_{ii}= d_i$ and $L_{ij}= -W_{ij}$. 
The Laplacian eigenvalues $0=\lambda_1,\lambda_2,\dots,\lambda_n$, usually sorted by magnitude, contain significant information about the network structure and its effect on the dynamical processes that could take place on it.

\alex{For undirected networks $W_{ij}=W_{ji}$, the symmetry of $L$ implies that all eigenvalues are real. In this case,} the multiplicity of $\lambda_1=0$ is equal to the number of connected components - separate parts - of the network. 
The second and last eigenvalues $\lambda_2$ and $\lambda_n$ provide bounds on the smallest and largest node degrees \cite{Fiedler}:
\begin{equation}\label{eq:bound_degrees}
d_{\min} \geq \frac{n-1}{n}\lambda_2 \hspace{1cm} d_{\max} \leq \frac{n-1}{n}\lambda_n.
\end{equation}
Finally, the leading eigenvectors of the Laplacian (e.g. Fiedler vector $v_2$) tend to have similar entries for nodes close to each other. They can be used to cluster the nodes. See for example \cite{von2007tutorial} for an introduction to such \emph{spectral clustering} techniques.

%
%
%
%

%

\alex{For large (directed) networks, it is more convenient to work with the spectral moments} of the Laplacian matrix
\begin{equation*}
\M_k = \frac{1}{n}\sum_{i=1}^n \lambda_i^k = \frac{1}{n} \tr (L^k)\,.
\end{equation*}
They are always real and also provide relevant information. The first moment is equal to the average node degree: $\M_1=\frac{1}{n}\sum_i d_i$. The following bounds on the quadratic mean degree were also shown in the Appendix of \cite{MauroyHendrickx:17}:
\begin{equation}
\label{eq:M2}
\max \prt{\M_1^2,\frac{\M^2_2}{2}}\leq \frac{1}{n}\sum_i (d_i)^2 \leq \M_2.
\end{equation}
We refer the interested reader to \cite{MauroyHendrickx:17, Fiedler, Preciado} for more detail.

\section{SPECTRAL NETWORK IDENTIFICATION}
\label{sec:eigenvalue}

\subsection{Problem statement}


We consider a network of $n$ identical units interacting through a diffusing coupling and forced by external inputs $u:\mathbb{R} \to \mathbb{R}^p$. Each unit is described by $m$ states $x_k \in \mathbb{R}^m$ which evolve according to the LTI dynamical system
\begin{eqnarray}
\label{eq:lin_dyn}
\dot{x}_k  & = & A x_k + B \sum_{j=1}^n w_{kj} (y_k-y_j) + D_k u(t) \\
y_k & = & C^T x_k
\end{eqnarray}
with $A\in \mathbb{R}^{m \times m}$, $B\in \mathbb{R}^{m}$, $C \in \mathbb{R}^{m}$, and $D_k \in \mathbb{R}^{m \times p}$. We assume that $A,B,C$ are identical for all units, but make no assumption on $D_k$, which can differ from one unit to another. For instance, the inputs could be applied to a small number of units. \revjh{We will refer to solutions of \eqref{eq:lin_dyn} as \emph{trajectories} of the system.}

Denoting by $X=[x_1 \dots x_n]^T \in \mathbb{R}^{mn}$ the vector which collects all the states, we obtain the dynamics
\begin{equation}
\label{eq:whole_syst}
\dot{X}=K X + D u(t)
\end{equation}
with $K = I_n \otimes A - L \otimes BC^T \in \mathbb{R}^{mn \times mn}$ and $D = [D_1,\cdots,D_n]^T \in \mathbb{R}^{mn \times p}$. The matrix $I_n$ is the $n \times n$ identity matrix, $L \in \mathbb{R}^{n \times n}$ is the Laplacian matrix and $\otimes$ denotes the Kronecker product.

In the context of spectral network identification, our goal is the following: assuming that the local dynamics of the units and the external inputs are known (i.e. $A$, $B$, $C$, and $u$ are known), obtain the Laplacian eigenvalues $\lambda_k$ from $q \ll n$ partial local measurements in the network. These measurements are time series of the form $Q X(kT) \in \mathbb{R}^q$, where the measurement matrix $Q \in \mathbb{R}^{q \times mn}$ is a sparse matrix and $T>0$ is the sampling time. Typically, we will consider direct measurements at a small number of units, i.e. $Q^T = [e_{k_1} \cdots e_{k_q}]$ for some $k_1,\dots,k_q \in\{1,\dots,mn\}$, where $e_k$ is the $k$th unit vector. Note that the inputs are known but they are applied to unknown locations in the network (the matrix $D$ is not be known).

\begin{assumption}[Zero-order hold assumption]
We make the standing assumption that the sampling time $T$ is small enough so that the input can be considered as constant over every sampling period $[kT,(k+1)T]$.
\end{assumption}
\begin{assumption}[Single experiment]
	We assume that we can measure only one trajectory of the system.
\end{assumption}

Our method to solve the above spectral identification problem is divided in two steps. (1) Eigenvalues of $K$ are estimated through the Dynamic Mode Decomposition (DMD) method---and in particular its extension to systems with inputs \cite{Proctor}. (2) Laplacian eigenvalues are obtained by exploiting a one-to-one correspondence with the spectrum of $K$ \cite{MauroyHendrickx:17}.

\subsection{Estimation of the eigenvalues of $K$}
In this section, we extract the eigenvalues of $K$ from measurement data points. We show how these eigenvalues are related to particular matrices obtained from measured data and we use the DMD method.

\paragraph{Spectrum of $K$} The method is based on the following idea. \jmhnew{For arbitrary trajectories of the system, \revjh{i.e. solutions of \eqref{eq:lin_dyn}}, we can define the vectors} $Z(k)\in \mathbb{R}^{qN}$ compiling the measurements and inputs at periodic intervals $T$, i.e.
\begin{equation*}
{\small
Z(k)=\begin{bmatrix} QX(kT) \\  QX((k+1)T)  \\ \vdots \\ QX((k+N-1)T) \end{bmatrix},\,
U(k)=\begin{bmatrix} u(kT) \\  u((k+1)T)  \\ \vdots \\ u((k+N-1)T) \end{bmatrix}}
\end{equation*}
There is a \jmhnew{general} linear relation between the vectors $Z(k)$, $U(k)$ and the vector $Z(k+1)$. Moreover, this linear relation is invariant under shift of all measurement times and input times by a constant multiple of $T$, and its spectrum is related to the spectrum of the system matrix $K$. This result is summarized in the following proposition.



\begin{proposition}
	\label{prop:Gamma}
	
	Suppose that $mn = Nq$ for some integer $N$ and that the square matrix
	$$
	O_N = \left[\begin{array}{c}
	Q e^{K(N-1)T}\\
	Q e^{K(N-2)T}\\
	\vdots\\
	Q
	\end{array} \right]
	$$ is full rank \jmhnew{ (which is generically the case, and implies observability of the pair $(K,Q)$)}. Then there exist unique matrices $\Gamma \in \mathbb{R}^{Nq\times Nq}$ and $\Upsilon \in \mathbb{R}^{Nq\times Np}$ such that
	\begin{equation}
	\label{eq:Gamma}
	Z(k+1) = \Gamma Z(k) + \Upsilon U(k)
	\end{equation}
\jmhnew{holds for all possible trajectories of the system.}
Moreover, $\mu$ is an eigenvalue of $K$ with the corresponding eigenvector $w$ if and only if $\tilde{\mu} = e^{\mu T}$ is an eigenvalue of $\Gamma$ with the corresponding eigenvector$\tilde{w} = M \otimes (Q w)$, where $M$ is a diagonal matrix with diagonal entries $[1 \ \mu \ \mu^2 \ \cdots \ \mu^{N-1}]$. \hfill $\diamond$
\end{proposition}
The proof is given in the Appendix. 

According to Proposition \ref{prop:Gamma}, we can thus focus on the eigendecomposition of $\Gamma$ to obtain the spectrum of $K$. This is done through the DMD method for control.

\paragraph{Data matrices and DMD for control \cite{Proctor}} 

\jmhnew{We now turn our attention to the computation of $\Gamma$ from our observations of one trajectory. Under the conditions of Proposition \ref{prop:Gamma}, $\Gamma$ and $\Upsilon$ are the unique matrices to satisfy \eqref{eq:Gamma} for all trajectories of the system. Hence the idea of the method is to require our estimates of $\Gamma$ and $\Upsilon$ to satisfy $\eqref{eq:Gamma}$ for all available data-points. The obtained system of equations will admit (under the conditions of Proposition \ref{prop:Gamma}) $\Gamma,\Upsilon$ as unique solutions, unless the trajectory lies in a specific proper subspace. In this latter case, the system would be underdetermined (e.g. $u\equiv 0$, $X\equiv 0 $ for the most trivial example), which is not a generic situation.}

From a practical point of view, we define the data matrices $\bar{Z} = [Z(0) \ Z(1) \ \cdots] $, $\bar{Z}' = [Z(1) \ Z(2) \ \cdots]$, and $\bar{U} = [U(0) \ U(1) \ \cdots]$.
\jmh{The width of these matrices should be taken as large as the data collected allows.}
Following the original DMD method for input,  we rewrite \eqref{eq:Gamma} as
$\bar{Z}' =  [\Gamma \  \Upsilon] \begin{bmatrix}\bar{Z} \\ \bar{U} \end{bmatrix}$. 
Then the unknown matrices are given by
\begin{equation}
\label{gamma_pseudo}
[\Gamma \  \Upsilon] = \bar{Z}' \begin{bmatrix}\bar{Z} \\ \bar{U} \end{bmatrix}^\dagger
\end{equation}
where ${}^\dagger$ denotes the Moore-Penrose pseudoinverse. The DMD method therefore consists of (1) computing $\Gamma$, obtained with the $Nq$ first columns in \eqref{gamma_pseudo}, and (2) computing the eigenvalues and eigenvectors of $\Gamma$. We note that the pseudoinverse is generally not used, and is replaced by a singular value decomposition (see \cite{Proctor} for more details).

\begin{remark}
		1. \jmhnew{In practice, we will prefer to define columns of the data matrices with larger time steps $\Delta$ than the sampling time $T$, so that $Z(k)$ would contain $QX(kT), QX(kT+\Delta), QX(kT+2\Delta),\dots$ (where $\Delta$ is a multiple of $T$). Numerical tests suggest indeed that the spectrum of $\Gamma$ obtained in this case still coincides with the spectrum of $K$, and the algorithm actually yields more accurate results.} \\ 
		2. The number $N$ chosen to define the vectors $Z(k)$ could be different from the condition $Nq=nm$ imposed in Proposition 1. If $Nq>nm$, then the spectrum of $\Gamma$ will coincide with the spectrum of $K$, but will contain additional zero eigenvalues. If $Nq<nm$ (which is typically the case for the values $N$ chosen with large networks), then the two spectra are different, but our numerical simulations show that the eigenvalues of $\Gamma$ lie near clusters of eigenvalues of $K$ (see Section \ref{sec:large_moments}).
\end{remark}

\subsection{Estimation of Laplacian eigenvalues \jmhnew{and eigenvectors}}

Now we infer the Laplacian eigenvalues from the eigenvalues of the matrix $K$. We rely on previous works, which show that there exists a bijection between the spectra of $L$ and $K$ \cite{MauroyHendrickx:17}. In particular, we use the following result.
\begin{proposition}[{\cite[Proposition 2]{MauroyHendrickx:17}}]
	\label{prop:net_ident}
	If $\mu$ is an eigenvalue of $K$ and is not an eigenvalue of $A$, then
	\begin{equation}
	\label{eq:Laplacian_eigenval}
	\lambda = \frac{1}{C^T(A-\mu I_m)^{-1}B}
	\end{equation}
	is a Laplacian eigenvalue. \hfill $\diamond$
\end{proposition}
The proof can be found in \cite{MauroyHendrickx:17}. \alex{We note that $m$ distinct eigenvalues of $K$ yield the same Laplacian eigenvalue through \eqref{eq:Laplacian_eigenval}, but all Laplacian eigenvalues are obtained when the full spectrum of $K$ is considered.}

\jmh{We now show that we can also retrieve some information on the Laplacian eigenvectors, which will prove useful to apply certain clustering methods, see \ref{subsec:clustering}.}
\begin{proposition}
	\label{prop:eigenvec_L}
	Assume that $Q=Q_1 \otimes Q_2$, for some $Q_1 \in \mathbb{R}^{q_1 \times n}$ and $Q_2 \in \mathbb{R}^{q_2 \times m}$ (with $q=q_1 \, q_2$) \alex{and consider the eigenvectors $\tilde{w}$ and $v$ of $\Gamma$ and $L$, associated with the eigenvalues $\tilde{\mu}=e^{\mu T}$ and $\lambda$, respectively, where $\lambda$ and $\mu$ are related through \eqref{eq:Laplacian_eigenval}. Then,}
	\begin{equation}
	\label{eq:ratios}
	 \frac{[Q_1 v]_i}{[Q_1 v]_j} = \frac{[\tilde{w}]_{l_1 q+(i-1)q_1 + l_2}}{[\tilde{w}]_{l_1 q+(j-1)q_1 + l_2}} \quad \,,
	\end{equation}
	where $[\ ]_i$ denotes the $i$th component of the vector, with $i,j=1,\dots, q_1$, $l_1 = 0,\dots,N-1$, and $l_2 = 1,\dots,q_2$.
	 \hfill $\diamond$
\end{proposition}
\alexb{In particular, when $Q^T=[e_{k_1} \, \cdots \, e_{k_q}]$, that is, we measure distinctly the same state at some nodes, the result can be applied with $[Q_1 v]_i=v_{k_i}$.}\\
\begin{proof}
It follows from {\cite[Lemma 1]{MauroyHendrickx:17}} that $w = v \otimes v_A$, where $v_A$ is an eigenvector of $A$. Then, Proposition \ref{prop:Gamma} implies that the eigenvector of $\Gamma$ is of the form
$\tilde{w} = \Lambda \otimes (Q_1 v \otimes Q_2 w)$
and \eqref{eq:ratios} directly follows.	
\end{proof}

\subsection{Algorithm}

Our method for spectral identification of networks with inputs is summarized in the following algorithm.
\begin{algorithm}[h]
	\caption{Spectral identification of networks with inputs}
	\label{alg:lifting}
	\begin{algorithmic}[1]
		\Statex{\bf Input:} Measurements $\{QX(kT)\}$; inputs values $\{u(kT)\}$; shift sequences parameters $N$ and $\Delta$.
		\Statex{\bf Output:} Laplacian eigenvalues $\lambda_k$ and ratios $[Q_1 v_k]_i/[Q_1 v_k]_j$.
		\State Construct the data matrices $\bar{Z}$, $\bar{Z}'$, and $\bar{U}$; 
		\State DMD for control (1): Compute the matrices $\Gamma$ and $\Upsilon$ using \eqref{gamma_pseudo};
		\State DMD for control (2): Compute the eigenvalues $\tilde{\mu}_k$ and eigenvectors $\tilde{w}_k$ of $\Gamma$;
		\State $\mu_k = 1/T \log \tilde{\mu}_k$;
		\State Compute the Laplacian eigenvalues $\lambda_k$ using \eqref{eq:Laplacian_eigenval};
		\State Compute the ratios $[Q_1 v_k]_i/[Q_1 v_k]_j$ using \eqref{eq:ratios}.
	\end{algorithmic}
\end{algorithm}

\subsection{Example}
\label{subsec:example}

We consider a directed random Erd{\H o}s-R\'enyi network with $n=15$ nodes and a probability $0.3$ for any two vertices to be connected. The weights of the edges are uniformly randomly distributed between $0$ and $5$. The dynamics is given by \eqref{eq:lin_dyn} with the (known) matrices
\begin{equation*}
A = \begin{bmatrix}  -1  &  -2 \\
	1  &  -1 \end{bmatrix} \quad 
B = \begin{bmatrix} 1 \\ 2 \end{bmatrix} \quad 
C =  \begin{bmatrix} 1 \\ 1 \end{bmatrix} 
\end{equation*}
and the two inputs are sinusoidal signals with random amplitude and frequency (between $0$ and $1$). They are applied on the second state of nodes $3$ and $4$, i.e.  $$D_3=\begin{bmatrix} 0 & 0 \\ 1 & 0 \end{bmatrix}, \hspace{.2cm}D_4=\begin{bmatrix} 0 & 0 \\ 0 & 1 \end{bmatrix},$$
and $D_k=0$ for $k\neq\{3,4\}$). We measure the first state of nodes $1$ and $2$, i.e. the measurement matrix is $Q^T=[e_1 \ e_2]$, on a single trajectory of the system (sampling time $T=0.001$ over the time interval $[0,2]$). Using the proposed spectral network identification method (with parameters $N=20$, $\Delta=0.03$), we can \alexb{obtain a good approximation} of almost all Laplacian eigenvalues (Figure \ref{fig:example1}). This provides a good estimate of the first spectral moments $\mathcal{M}_1$ and $\mathcal{M}_2$, and of \jmh{the mean node degrees  $\mathcal{D}_1 = \mathcal{M}_1$ and  mean quadratic node degrees $\mathcal{D}_2$} through \eqref{eq:M2} (Table \ref{table1}). Using \eqref{eq:ratios}, we can also compute the ratio $[v_2]_1/[v_2]_2$ between the first and second components of the Fiedler vector. \jmh{The fact that it is negative tends to indicate that nodes $1$ and $2$ are in \quotes{opposite parts} of the network, see Section \ref{subsec:clustering} for more detail} \alex{on the use of the eigenvectors in the context of clustering.}

Note that more accurate results could be obtained if several trajectories associated with different initial conditions were measured.


\begin{figure}[h]
	\centering
	\includegraphics[width=0.6\linewidth]{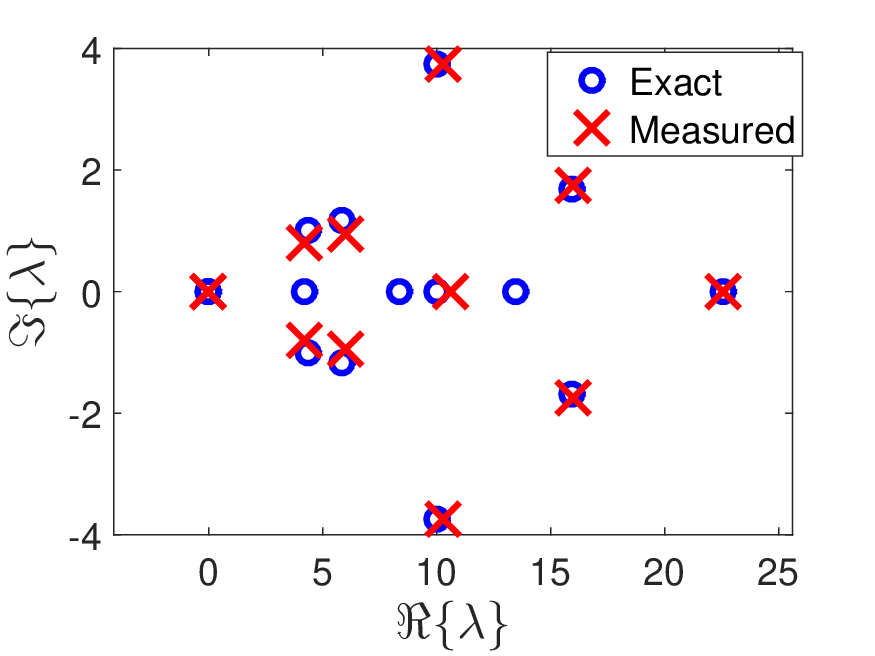}
	\caption{Almost all the Laplacian eigenvalues of a networks with $15$ nodes are estimated precisely with only two measurements (one state at two nodes).}
	\label{fig:example1}
\end{figure}

%
%
%
%

\begin{table}[h]
	\centering
\begin{tabular}{c|cccc}
	\hline 
 &$\mathcal{M}_1$	= $\mathcal{D}_1$ &   $\mathcal{M}_2$ & $\mathcal{D}_2$ & $v_2(1)/v_2(2)$
 \\
	\hline 
Exact& 	8.73 & 110.33 &  105.38 & -2.28 \\
Estimated &	8.85 & 117.88 & $\leq$117.88 & -2.06
\end{tabular}
\caption{}
\label{table1}
\end{table}

\section{APPLICATIONS TO LARGE NETWORKS}
\label{sec:large_moments}

In the case of large networks, it is typically impossible to infer the full spectrum of the Laplacian eigenvalues. In this situation, we rather focus on (1) the first spectral moments of the Laplacian matrix (mean node degree), (2) the algebraic connectivity $\lambda_2$ and the largest eigenvalue (in real part) $\lambda_n$ (minimum and maximum node degree), and (3) the leading spectral eigenvectors (clustering).

\subsection{Mean node degree}

\jmh{The mean node degree and mean quadratic node degree} can be obtained by considering the spectral moments of the Laplacian matrix. In the case of large networks, the Laplacian eigenvalues cannot be obtained exactly. However, we have observed experimentally that \jmhnew{they lie inside or near the convex hull $\mathcal{S}$ of a few values obtained when applying the method described in the previous section with $N \ll nm/q$.}
\jmh{We therefore estimate the spectral moments by computing the moments of area of this convex hull.}
In particular, we have
\begin{equation}
\label{area_moment_K}
\mathcal{M}_1 \approx \displaystyle  \frac{\int_{\mathcal{S}_j} x \, dx dy}{\int_{\mathcal{S}_j} dx dy}, \quad
\mathcal{M}_2 \approx \displaystyle \frac{\int_{\mathcal{S}_j} y^2-x^2 \, dx dy}{\int_{\mathcal{S}_j} dx dy}
\end{equation}
\jmh{with $(x,y)=(\Re\{\lambda\},\Im\{\lambda\})$.} 
We note that this result can be generalized to the case where several clusters of eigenvalues require to consider several convex hulls (see \cite{MauroyHendrickx:17} for more details).

In the following example, we consider a directed random network of $100$ nodes, with a normal distribution of the number of edges at each node (mean: $10$, standard deviation: $5$). The weights of the edges are uniformly randomly distributed between $0$ and $0.1$. The dynamics are the same as in Example \ref{subsec:example}, but with $5$ sinusoidal inputs (with random amplitude and frequency between $0$ and $1$) applied to $5$ different nodes. We measure one state of one node only (sampling time $T=0.02$ over the time interval $[0,20]$). Using our method for spectral identification of networks with inputs (with parameters $N=20$, $\Delta=0.2$), we can obtain a convex hull that approximates the location of the Laplacian eigenvalues (Figure \ref{moment_large}). This provides accurate values of the spectral moments, and therefore a fair estimate of the mean node degrees (Table \ref{table2}). \alex{Note that the error on the mean node degree ($0.02$) is very small compared to the standard deviation of the nodes degree distribution $\sqrt{\mathcal{D}_2-\mathcal{D}_1^2} = 0.26$.} We can also divide the (estimated) average node degree by the average value of the \jmh{edge} weights ($0.05$ in our case), and we obtain an estimation of the average number of edges per node.

\begin{figure}[h]
	\centering
	\includegraphics[width=0.7\linewidth]{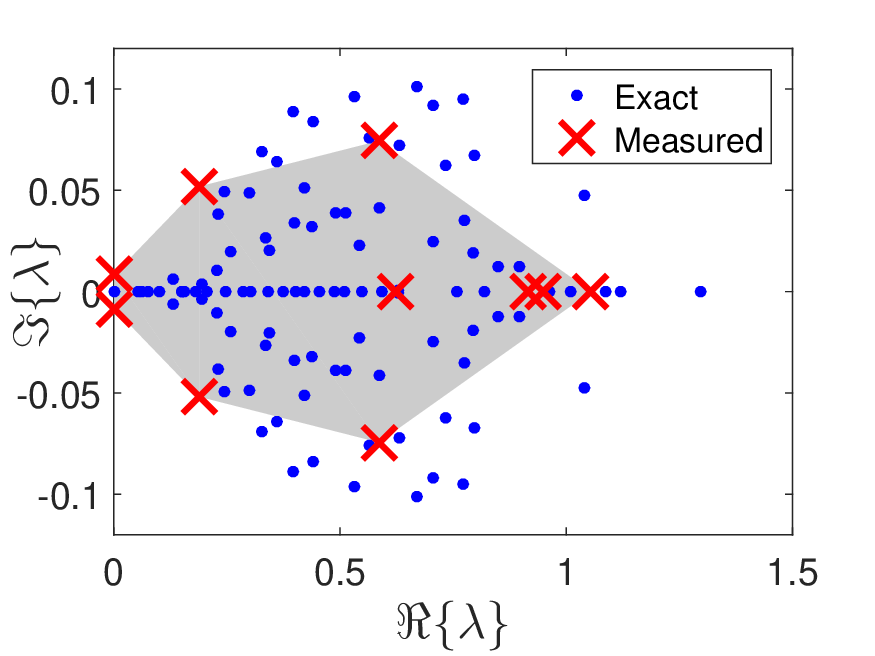}
	\caption{The convex hull of the estimated eigenvalues  approximates the location of the Laplacian eigenvalues.}
	\label{moment_large}
\end{figure}

%
%
%
%

\begin{table}[h]
	\centering
	\begin{tabular}{c|cccc}

& $\mathcal{M}_1=\mathcal{D}_1$ & avg. number of edges & $\mathcal{M}_2$ & $\mathcal{D}_2$\\ \hline
Exact &0.51 & 10.22 & 0.33 & 0.33\\
Estimated & 0.49 & 9.81 & 0.29 & $[0.26 , 0.29]$
	\end{tabular}
		\caption{}
	\label{table2}
\end{table}

\subsection{Minimum and maximum node degrees}

In this example, we focus on the Laplacian eigenvalues $\lambda_2$ and $\lambda_n$, which provide bounds on the minimum and maximum node degrees according to \eqref{eq:bound_degrees}, in the case of undirected networks. We consider the \jmhnew{undirected and unweighted} Polblogs network \cite{Adamic}. This network originally describes links between $1224$ US political blogs. We attach one state to each node, \jmhnew{representing the node opinion,} and consider a consensus dynamics of the form \eqref{eq:lin_dyn} with $A=-1$, $B=0.1$, and $C=1$. As an external influence, three sinusoidal inputs with random amplitude between $0$ and $0.1$, and random frequency between $0$ and $1$ are applied to the network. Each input directly affects one third of the nodes. We measure the state of node $1$, whose degree is $26$. \jmhnew{(Our sampling time is $T=0.01$ over the time interval $[0,10]$, and we use the parameters $N=50$, $\Delta=0.1$ in our method.)}


Using this single measurement, the spectral identification method provides a good estimate of the Laplacian spectrum range (Figure \ref{fig:min_max}). The eigenvalue $\lambda_2$ is overestimated, but this result still gives the order of magnitude of the minimum node degree. \jmh{Note that we actually estimate much more accurately the first eigenvalue  $\lambda_1=0$, which is unfortunately less interesting in this context}. More importantly, the estimate of \jmhnew{$\lambda_n$} is very accurate, a result which allows to obtain a very close bound for the maximum node degree (Table \ref{table3}). \jmhnew{With the measure of one node only, one can therefore deduce the presence of a highly connected node and estimate how influent it is. }

\begin{figure}[h]
	\centering
	\includegraphics[width=0.6\linewidth]{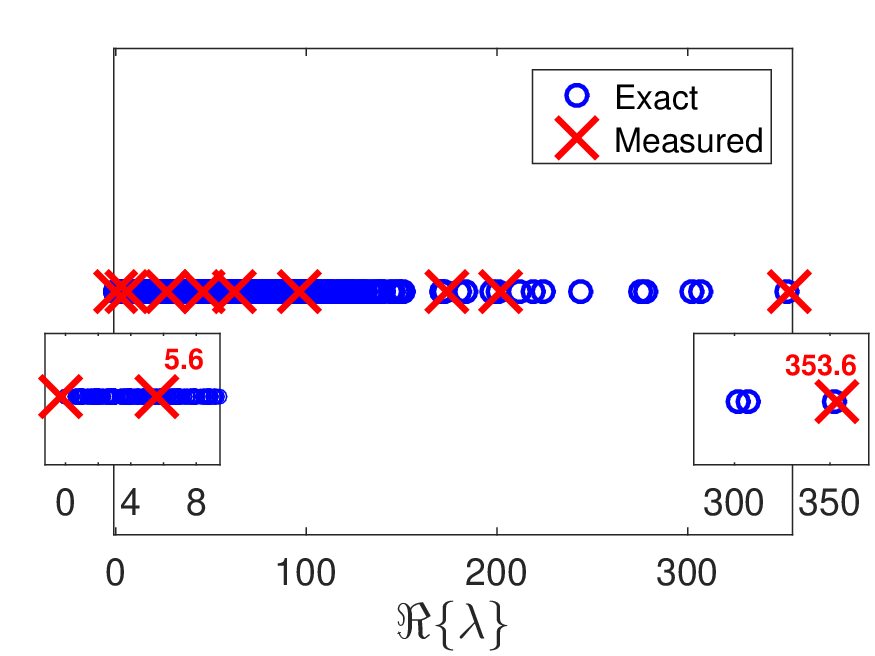}
	\caption{We obtain a good estimate of the range of the Laplacian spectrum, which provides bounds on the minimum and maximum node degrees. The two insets show the eigenvalues near $\lambda_2$ and near $\lambda_n$.}
	\label{fig:min_max}
\end{figure}

\begin{table}[h]
	\centering
	\begin{tabular}{c|cccc}

& $\lambda_2$& $\lambda_n$& $d_{\min}$& $d_{\max}$\\ \hline
Exact & 0.17 & 352.05 & 1 & 351\\
Estimated & 5.61 & 353.56 & $>5$ & $<353$
	\end{tabular}
	\caption{}
	\label{table3}
\end{table}

%
%
%
%

\subsection{Clustering}
\label{subsec:clustering}

We will now use the spectral identification framework in the context of clustering. Our approach is very similar to the technique proposed in \cite{Raak} to partition a power grid using the so-called Koopman mode analysis on measurements of the full network. Proposition \ref{prop:eigenvec_L} implies that we can compute the ratios between components of leading Laplacian eigenvectors if we use measurements at the corresponding nodes, a result which can be used to decide whether nodes are part of the same cluster.

We consider a \jmh{random graph with planted partitions, consisting of $3$ planted clusters of $50$ nodes each. Any two nodes of the same cluster are connected with independent probability $0.3$, but this probability drops to $0.05$ if the two nodes belong to different clusters.}
The weights of the edges are uniformly randomly distributed between $0$ and $0.1$. The dynamics are the same as in Example \ref{subsec:example}, with $3$ sinusoidal inputs with random amplitude and frequency (between $0$ and $1$), each of which is applied to one node in a distinct cluster. We measure one state for $5$ nodes in each cluster (sampling time $T=0.02$ over the time interval $[0,20]$) and use our method (with parameters $N=2$, $\Delta=0.2$). The ratios $[v_2]_j/[v_2]_1$ and $[v_3]_j/[v_3]_1$ (where $1$ and $j$ are measured nodes) are computed by using \eqref{eq:ratios} and are depicted in Figure \ref{fig:clustering}. The obtained result correctly retrieve the three clusters.

\begin{figure}[h]
	\centering
	\includegraphics[width=0.7\linewidth]{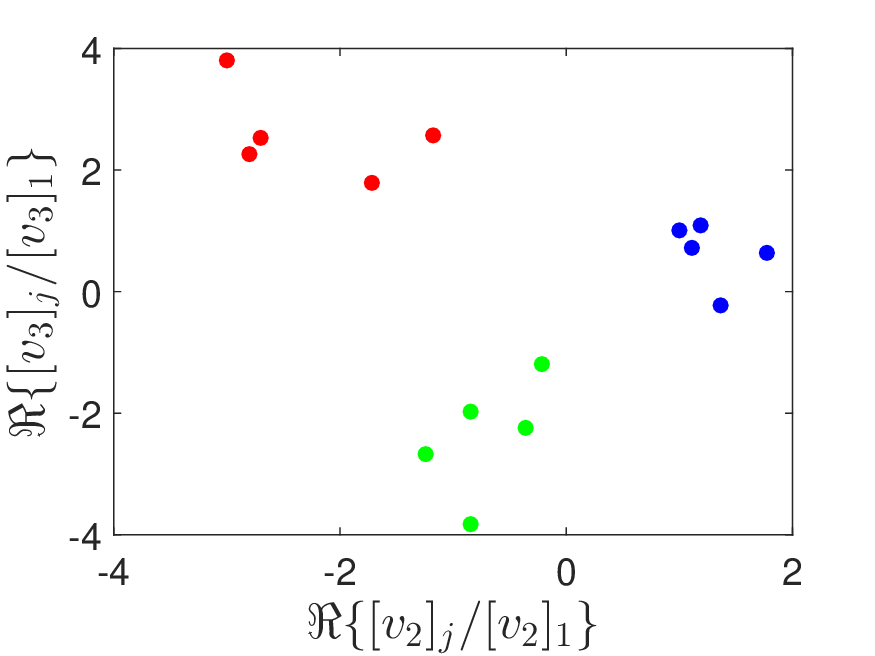}
	\caption{The ratios between component of the leading Laplacian eigenvector are computed with the spectral identification method and correctly infer the three clusters. Each dot corresponds to a ratio $[v_2]_j/[v_2]_1$ and $[v_3]_j/[v_3]_1$ and is colored according to the exact cluster of node $j$.}
	\label{fig:clustering}
\end{figure}

\section{CONCLUSIONS AND PERSPECTIVES}
\label{sec:conclusions}

We have extended our spectral network identification framework to allow for external inputs in the linear case. Crucially, this allows to \alex{consider networked systems that do not synchronize} and removes the need for using multiple trajectories. It is also relevant in applications where external inputs are present and cannot be avoided.
In this paper, we \jmh{have demonstrated the use of our framework} to estimate the mean, minimum, and maximum node degrees of the network, from measurements at a small number of nodes. Moreover, we have also shown that we can infer some information on the leading eigenvectors and use this result for network clustering.

One main challenge is the possible extension of these results to nonlinear interconnected systems that do not necessarily converge (unlike what was assumed in \cite{MauroyHendrickx:17}), to other sorts of coupling, or to networks with heterogeneous node dynamics. Some additional assumptions will be required in the latter case, since it was shown in \cite{MauroyHendrickx:17} that recovering the Laplacian spectrum from that of $K$ is in general impossible.

\alexb{Another line of potential further research concerns the optimization of the numerical algorithms. 
There might be more accurate ways of estimating the eigenvalues of the Laplacian based on those of $K$ when the latter are not fully accurate. 
Moreover, it is not clear that numerical methods are still efficient when (process or measurement) noise is added to the system. How numerical methods (and the underlying theoretical framework) should be adapted to properly handle noise is an open question.}


One could also explore the link with classical system identification. The eigenvalues of $K$ \jmh{defined in \eqref{eq:whole_syst}}
that we compute with the DMD algorithm are indeed (generically) also the poles of the transfer functions from the input to the measured nodes. One could thus in principle obtain them by identifying these transfer functions in a \quotes{classical} way, which would also naturally allow for noise in the system. However, classical identification techniques were designed for other challenges, and do not necessarily minimize the errors on the poles, so it is not clear that they would give better results, especially since the order of the transfer functions could be very large: \jmh{treating the example of Section \ref{subsec:example}, which involves 15 nodes and 2 states per node, would for instance require identifying a transfer function of order 30}.

Finally, we recall that our end-goal is to identify properties of the network from sparse measurements. Reconstructing the Laplacian spectrum is one way of attacking this problem, but not a goal on its own. There could be other relevant quantities containing information about the graph and that could be deduced from the trajectories. 

\appendix

\begin{proofof}\emph{Proposition \ref{prop:Gamma}}.
	
	We first prove two lemmas. Note that the proof is somewhat technical and given here for reviewing purpose. It will be omitted in the final version.
	\begin{lemma}
		\label{lemma_Julien1}
		Consider the discrete time-system 
		\begin{align}
		x(t+1) & = K x(t) + D u(t)\\
		z(t) & = Q x(t),
		\end{align}
		with $x\in \mathbb{R}^{nm}$, $K \in \mathbb{R}^{nm\times nm}$, $Q\in \mathbb{R}^{q\times nm}$ and $y(t)\in \mathbb{R}^q$, and suppose the pair $[Q,K]$ is observable.
		Suppose in addition that $mn = Nq$ for some integer $N$ and that 
		$$
		O_N = \left[\begin{array}{c}
		Q K^{N-1}\\
		Q K^{N-2}\\
		\vdots\\
		Q
		\end{array} \right]
		$$
		is full rank. Let 
		$$
		Z(t) = \left[\begin{array}{c}
		z(t) \\
		z(t-1)
		\vdots\\
		z(t-N+1)
		\end{array} \right], \hspace{.2cm} U(t) = \left[\begin{array}{c}
		u(t) \\
		u(t-1)
		\vdots\\
		u(t-N+1)
		\end{array} \right], 
		$$
		be a vector compiling the last $N$ observations. 
		Then, there exist unique matrices $\Gamma$ and $\Upsilon$ such that
		\begin{equation}\label{eq: Z=GZUU}
		Z(t+1) = \Gamma Z(t) + \Upsilon U(t),
		\end{equation}
		for all $t$.
		Moreover, $w$ is an eigenvector of $K$ with eigenvalue $\tilde{\mu} \neq 0$ if and only if
		$$
		\tilde{w}= 
		\prt{\begin{array}{c}
			\mu^{N-1}Qv\\
			\mu^{N-2}Qv\\
			\vdots\\
			Qv
			\end{array}
		}
		$$
		is an eigenvector of $\Gamma$ with the same eigenvalue $\tilde{\mu}$.
	\end{lemma}
	\begin{proof}
		We first prove \eqref{eq: Z=GZUU}. The last $(N-1)q$ lines of the equations follow directly from the definition of $\Gamma$ and the 0 lines in $\Upsilon$, so we focus on the $q$ first ones, which are
		$$
		z(t+1) = QK^N O_N^{-1} Z(t) + \tilde \Upsilon U(t),
		$$
		with $\tilde \Upsilon$ the first $q$ rows of $\Upsilon$.
		Let $\tilde x = x(t-N+1)$. Observe that 
		$$
		Z(t) = \left( 
		\begin{array}{c} 
		QK^{N-1}\tilde x + Q\sum_{k=1}^{N-1} K^{k-1}Du(t-N+k) \\
		\vdots
		\\
		QK\tilde x + QDu(t-N+1) 
		\\
		Q\tilde x
		\end{array}
		\right)
		$$
		Therefore, it follows that
		$$
		\tilde x = O_N^{-1}Z(t) - M_1 U(t),
		$$
		for some matrix $M_1$. Since $x(t+1) = Kx(t) + Du(t) = K(K^{N-1}\tilde x +\sum_{k=1}^{N-1} K^{k-1}Du(t-N+k)) + Du(t)$, this implies
		$$
		z(t+1) = Qx(t+1) = QK^N O_{N}^{-1}Z(t) + \tilde \Upsilon U(t)
		$$
		for some $\tilde \Upsilon$.
		
		Let us now move to the eigenvalues. Due to the structure of $\Gamma$, the equation $\Gamma w = \mu w$ is satisfied if and only if (i) $w^T = (\mu^{N-1} z^T, \mu^{N-2}z^T,\dots, z^T)^T$ for some $z$ of dimension $q$, and (ii)
		\begin{equation}\label{eq:condition_eigenvalue}
		\mu^{N}z^T = QK^NO_N^{-1} (\mu^{N-1} z^T, \mu^{N-2}z^T,\dots, z^T)^T.
		\end{equation}
		Hence $\mu$ is an eigenvalue of $\Gamma$ if and only if \eqref{eq:condition_eigenvalue} is satisfied for some $z$.
		
		Suppose first that $\mu$ is an eigenvalue of $K$ with eigenvector $x$, and let $z= Qx$. Observe that 
		
		$QK^NO_N^{-1}(\mu^{N-1} z^T, \mu^{N-2}z^T,\dots, z^T)^T$
		\begin{align*}
		&= QK^NO_N^{-1}(\mu^{N-1} (Q x)^T, \mu^{N-2}(Qx)^T,\dots, (Q x)^T)^T\\
		&= QK^NO_N^{-1}((Q K^{N-1} x)^T, (QK^{N-2}x)^T,\dots, (Q x)^T)^T\\
		&=QK^NO_N^{-1} O_N x\\
		&=Q K^Nx  =Q\mu^Nx,  
		\end{align*}
		so \eqref{eq:condition_eigenvalue} is satisfied, and $\mu$ is an eigenvalue of $\Gamma$.
		
		
		Suppose now on the other hand that $\mu$ is an eigenvalue of $\Gamma$ so that \eqref{eq:condition_eigenvalue} is satisfied for some $z$, and define 
		\begin{equation}\label{eq:defx}
		x = O_N^{-1} (\mu^{N-1}z^T, \mu^{N-2}z^T,\dots z^T)^T.
		\end{equation}
		Multiplying this equation by $O_N$ leads to $QK^kx = \mu^k z$ for every $k=0,\dots, N-1$. 
		It also follows from \eqref{eq:condition_eigenvalue} that 
		\begin{align*}
		\mu^N z^T &= QK^NO_N^{-1} (\mu^{N-1} z^T, \mu^{N-2}z^T,\dots, z^T)^T\\
		&= QK^Nx.
		\end{align*}
		Observe now that from \eqref{eq:defx}
		\begin{align*}
		\mu x &= O_N^{-1} (\mu^{N}z^T, \mu^{N-1}z^T,\dots,\mu z^T)^T\\
		& = O_N^{-1} ((QK^{N}x^)T, QK^{N-1}x^T,\dots,QK x^T)^T\\
		& = O_N^{-1} ((QK^{N-1}(Kx))^T, (QK^{N-2}Kx)^T,\dots,(Q (K x))^T)^T \\
		&= O_N^{-1} O_N (Kx) \\
		&= Kx,
		\end{align*}
		hence $x$ is an eigenvector of $K$ with eigenvalue $\mu$.
		
	\end{proof}
	
	\begin{lemma}
		\label{lemma_Julien2}
		\jmhnew{If $O_N$ is invertible, then $\Gamma, \Upsilon$ are the unique matrices for which \eqref{eq: Z=GZUU} holds for all possible trajectories. }
		\end{lemma}
	\begin{proof}
		Suppose there exist $\Gamma',\Upsilon'$ such that 
		$$
		Z(t+1) = \Gamma' Z(t) + \Upsilon' U(t).
		$$
		for every trajectories. Then there would hold
		\begin{equation}\label{eq:DGDU=0}
		0 = \Delta_\Gamma Z(t) + \Delta_\Upsilon U(t)
		\end{equation}
		for every trajectory with $\Delta_\Gamma = \Gamma'-\Gamma$ and $\Delta_\Upsilon = \Upsilon'-\Upsilon$. Let us consider a trajectory $\tilde x$ leaving from $\tilde x(0) = x_0$, with zero input at and after 0, so that $U(N-1) = 0$.
		There holds $\tilde x(t) = K^t x_0$, hence 
		\begin{align*}
		Z(N-1) & = ((QK^{N-1}x_0)^T,(QK^{N-2}x_0)^T,\dots Qx_0 )^T \\
		& = O_N x_0.
		\end{align*}
		It follows from \eqref{eq:DGDU=0} and $U(N-1)=0$ that
		$$
		0 = \Delta_\Gamma Z(N-1) = \Delta_\Gamma O_N x_0
		$$
		for every $x_0$. Since $O_N$ is invertible, this implies $\Delta_\Gamma = 0$. \jmhnew{Hence \eqref{eq:DGDU=0} becomes $\Delta_\Upsilon U(t)=0$ for all possible trajectories. Since $U(t)$ can be arbitrarily selected, this implies $\Delta_\Upsilon = 0$, and we have thus $\Gamma'=\Gamma$ and $\Upsilon'=\Upsilon$}
		
		
	\end{proof}
	
	\emph{Proof of Proposition \ref{prop:Gamma}}
	
	We note that the solution to \eqref{eq:whole_syst} is given by
	\begin{equation*}
	X(t) = e^{Kt} X(0) + \int_0^T e^{K(t-\tau)} D u(\tau) d\tau
	\end{equation*}
	and, under the zero-order hold assumption, we obtain 
	\begin{equation}
	\label{eq:tilde_K}
	X(kT) = \tilde{K} X((k-1)T) + \tilde{D} u(kT)
	\end{equation}
	with $\tilde{K}=e^{KT}$ and $\tilde{D}=K^{-1}(e^{KT}-I) D$. Given the assumptions, Lemmas \ref{lemma_Julien1} and \ref{lemma_Julien2} can be applied to the system \eqref{eq:tilde_K}. This implies that $\Gamma$ exists and is unique. Moreover, its eigenvalue $\tilde{\mu}$ is also an eigenvalue of $\tilde{K}$, so that $\tilde{\mu}= e^{T \mu}$ where $\mu$ is an eigenvalue of $K$. Similarly, $\tilde{w}$ is an eigenvector of $\Gamma$, where $w$ is an eigenvector of $\tilde{K}$ and $K$. This concludes the proof.

\end{proofof}

\bibliographystyle{plain}

\begin{thebibliography}{100}
	
	\bibitem{Adamic}
	L. A. Adamic and N. Glance, The political blogosphere and the 2004 US Election, {\em Proceedings of the WWW-2005 Workshop on the Weblogging Ecosystem} (2005).
	
	\bibitem{Chiuso}
	A. Chiuso and G. Pillonetto, A. Chiuso and G. Pillonetto,
	{\em Automatica}, 48 (2012), pp.~1553--1565
	
	\bibitem{Dankers}
	A. Dankers, P.M.J. Van den Hof, X. Bombois, and P.S.C. Heuberger., {Identification of dynamic models in complex networks with prediction error methods: Predictor input selection}, {\em IEEE
	Transactions on Automatic Control}, 61(4), (2016), pp.~937--952.
	
	
	\bibitem{Fiedler}
	M.~Fiedler, Algebraic connectivity of graphs, {\em Czechoslovak
	mathematical journal}, 23 (1973), pp.~298--305.
	
	\bibitem{Franceschelli}
	M.~Franceschelli, A.~Gasparri, A.~Giua, and C.~Seatzu, 
	  {Decentralized estimation of Laplacian eigenvalues in multi-agent systems},
	  {\em Automatica}, 49 (2013), pp.~1031--1036.
			
	\bibitem{Gevers}
	M. Gevers, A. Bazanella, and A. Parraga,  {On the identifiability of dynamical networks}, {\em To appear in the Proceedings of IFAC World Congress}, 2017.
	
	
	\bibitem{net_ident_opti}
	T.~He, X.~Lu, X.~Wu, J.~Lu, and W.~X. Zheng,  Optimization-based
	  structure identification of dynamical networks, {\em Physica A: Statistical
	  Mechanics and Its Applications}, 392 (2013), pp.~1038--1049.
	
	\bibitem{MauroyHendrickx:17}
	 A. Mauroy and J.~M. Hendrickx, Spectral identification of networks using sparse measurements, {\em SIAM Journal of Applied Dynamical Systems},  (2017), pp.~479-513.
	
	
	\bibitem{Sauer_net_ident}
	 D.~Napoletani and T.~D. Sauer, Reconstructing the topology of
		sparsely connected dynamical networks, {\em Physical Review E}, 77 (2008),
	p.~026103.
	
	\bibitem{Pikovsky_net_ident}
	 Z.~Levnaji{\'c} and A.~Pikovsky, Untangling complex dynamical
	  systems via derivative-variable correlations, {\em Scientific reports}, 4 (2014).
	
	\bibitem{Preciado}
	V.~M. Preciado, A.~Jadbabaie, and G.~C. Verghese, {Structural
			analysis of Laplacian spectral properties of large-scale networks}, {\em IEEE
	Transactions on Automatic Control}, 58 (2013), pp.~2338--2343.
	
	\bibitem{Proctor}
	J. L. Proctor, S. L. Brunton, and J. N. Kutz, Dynamic mode decomposition with control, {\em SIAM Journal on Applied Dynamical Systems}, 15(1), (2016), pp.~142--161
	
	\bibitem{Raak}
	 F.~Raak, Y.~Susuki, and T.~Hikihara, {Data-driven partitioning of
	  power networks via Koopman mode analysis}, {\em IEEE Transactions on Power
	  Systems}, 31, (2016), pp.~2799-2808
	
	\bibitem{Ren_net_ident}
	J.~Ren, W.-X. Wang, B.~Li, and Y.-C. Lai,  Noise bridges dynamical
	  correlation and topology in coupled oscillator networks, {\em Physical Review
	  Letters}, 104 (2010), p.~058701.
	
	\bibitem{Banaszuk_eigenvalues}
	 T.~Sahai, A.~Speranzon, and A.~Banaszuk,  Hearing the clusters of a
	  graph: {A} distributed algorithm, {\em Automatica}, 48 (2012), pp.~15--24.
	

	\bibitem{Schmid}
	 P.~J. Schmid,  Dynamic mode decomposition of numerical and
		experimental data, {\em Journal of Fluid Mechanics}, 656 (2010), pp.~5--28.
	
	
	
	\bibitem{Timme_review}
	 M.~Timme and J.~Casadiego, Revealing networks from dynamics: an
	  introduction, {\em Journal of Physics A: Mathematical and Theoretical}, 47 (2014),
	  p.~343001.
	  
	  \bibitem{Timme2_net_ident}
	   S.~G. Shandilya and M.~Timme,  Inferring network topology from
	    complex dynamics, {\em New Journal of Physics}, 13 (2011), p.~013004.
    
    	
    \bibitem{Kibangou}
    T.-M.-D. Tran and A.~Y. Kibangou, Distributed estimation of
      {L}aplacian eigenvalues via constrained consensus optimization problems,
      {\em Systems \& Control Letters}, 80 (2015), pp.~56--62.
    
	  
	
	\bibitem{Tu}
	J.~H. Tu, C.~W. Rowley, D.~M. Luchtenburg, S.~L. Brunton, and J.~N. Kutz,
	{On dynamic mode decomposition: Theory and applications}, {\em Journal of
	Computational Dynamics}, 1 (2014), pp.~391 -- 421.
	
	
	\bibitem{von2007tutorial}
	 U. Von Luxburg, { A tutorial on spectral clustering}, {\em Statistics and computing}, 17(4), (2007), pp.~395--416.
	
	\bibitem{net_ident_compressed_sens}
	 W.-X. Wang, R.~Yang, Y.-C. Lai, V.~Kovanis, and M.~A.~F. Harrison,
	  Time-series--based prediction of complex oscillator networks via compressive
	  sensing, {\em Europhysics Letters}, 94 (2011), p.~48006.
	
	\bibitem{Yu_Parlitz_steady_state}
	 D.~Yu and U.~Parlitz, Driving a network to steady states reveals its
	  cooperative architecture, {\em Europhysics Letters}, 81 (2008), p.~48007.
	
	\bibitem{Yu_net_estimation}
	D.~Yu, M.~Righero, and L.~Kocarev, Estimating topology of networks,
	  {\em Physical Review Letters}, 97 (2006), p.~188701.


	
\end{thebibliography}

\end{document}